\documentclass[preprint,12pt]{elsarticle}
\usepackage{etoolbox}
\makeatletter
\def\ps@pprintTitle{%
	\let\@oddhead\@empty
	\let\@evenhead\@empty
	\def\@oddfoot{\reset@font\hfil\thepage\hfil}
	\let\@evenfoot\@oddfoot
}
\makeatother
\usepackage{graphics}
\usepackage[labelfont=bf]{caption}
\usepackage[table,xcdraw]{xcolor}
\usepackage{amsmath, amsthm, amssymb}
\usepackage{natbib}
\usepackage[colorlinks=true]{hyperref}
\usepackage{lscape}
\usepackage[utf8]{inputenc} 
\usepackage{amsthm}

\usepackage{enumerate}
\usepackage{mathrsfs}
\usepackage{setspace}
\usepackage{multirow}
\usepackage{graphicx}
\usepackage{amsfonts}
\usepackage{verbatim}
\usepackage{amssymb}
\usepackage{amsthm}
\usepackage{diagbox}
\usepackage{booktabs}
\usepackage{array}
\usepackage{multirow,bigstrut,bigdelim}
\newtheorem{thm}{Theorem}[section]

\theoremstyle{plain}
\newtheorem{defi}{Definition}[section]
\newtheorem{ex}{Example}[section]

\theoremstyle{remark}

\numberwithin{equation}{section}
\usepackage[mathscr]{euscript}
\usepackage{geometry} 
\usepackage{graphicx,lscape} 
\usepackage{booktabs}
\usepackage{array}
\usepackage{paralist} 
\usepackage{verbatim}
\usepackage{subfig} 
\usepackage{orcidlink}
\biboptions{comma,round,authoryear}
   
\begin{document}
\begin{frontmatter}
		\title{\textbf{On Some Multivariate Extensions to Zenga Curve: Properties and Applications}}		
    \author{Shifna, P. R \orcidlink{0009-0002-4240-6941}\corref{cor1}}
	\ead{shifnarazack92@gmail.com}
	\author{S. M. Sunoj \orcidlink{0000-0002-6227-1506}}
	\ead{smsunoj@cusat.ac.in}
	
	
	\cortext[cor1]{Corresponding author}
		
\address{Department of Statistics\\Cochin University of Science and Technology\\ Cochin 682 022, Keralam, INDIA}

\begin{abstract}
Measures of inequality are often limited in their ability to capture multidimensional aspects that arise from the joint distribution of multiple socio-economic variables. In this paper, we develop bivariate extensions of the Zenga inequality measure using bivariate quantile functions. We propose new bivariate Zenga surfaces and study their theoretical properties. A vector-valued bivariate Zenga curve is also introduced to provide a more detailed characterization of inequality. A non-parametric estimator is proposed and methods are evaluated through simulation studies and applied to the analysis of digital inequality across countries using indicators such as broadband penetration and digital literacy. The results highlight the effectiveness of the proposed framework in capturing multidimensional inequality.
\end{abstract}
	\end{frontmatter}
\section{Introduction}
  The measurement of economic inequality has remained a central theme in economics and statistics for over a century, beginning with the pioneering contributions of \citet{lorenz1905methods} and \citet{gini1912variabilita}. The Lorenz curve and the associated Gini index have since become the most widely used tools for analyzing the concentration of a single variable such as income. Despite their popularity, these classical measures are inherently aggregate in nature and often fail to capture structural differences across distributions, particularly when inequality varies across different segments of the population.
\par In response to these, \cite{zenga1984proposta} proposed a novel inequality curve that offers a distinct alternative perspective to classical measures such as the Gini index. The fundamental idea is to evaluate inequality locally by comparing, at each population proportion $u \in (0,1)$, the position of the corresponding quantile in the original distribution with that in an income-weighted distribution. Let $X>0$ be an absolutely continuous random variable with cumulative distribution function $F(x)$ and finite mean $\mu$. The first moment distribution is defined as
\[
F^{(1)}(x) = \frac{1}{\mu}\int_0^x {t dF(t)}\]
with density
\[
f^{(1)}(x) = \frac{x f(x)}{\mu}.
\]
This distribution assigns weights proportional to income and represents the cumulative share of total income received by individuals with income less than or equal to $x$.
Let $Q(u) = F^{-1}(u)$ and $Q^{(1)}(u) = {F^{(1)}}^{-1}(u)$ denote the corresponding quantile functions. Then the Zenga curve is defined as
\begin{equation}\label{1.1}
  Z(u) = 1 - \frac{Q(u)}{Q^{(1)}(u)}, \quad 0<u<1.  
\end{equation}
$Z(u)$ captures the relative disparity between the original distribution and its income-weighted counterpart at each quantile level. The curve is bounded such that
$0 \leq Z(u) < 1$ for all $u\in (0,1)$.
It achieves the value of 0 only in the case of perfect equality, $i.e.$, a degenerate distribution where every individual has the same income. Graphically, a Zenga curve that is high for low values of $u$ indicates significant inequality among the poor, while a curve that increases towards $u=1$ suggests the gap between the rich and poor widens across the distribution. A high, flat curve indicates persistent inequality throughout.
The Zenga index is obtained by integrating the zenga curve,
\[
\xi = \int_0^1 {Z(u) du},
\]
providing a normalized and scale-invariant measure of inequality.  Unlike the Gini index, which is more sensitive to changes in the middle, the Zenga index offers greater sensitivity to disparities at the extremes, particularly the lower tail and yields a constant curve for the lognormal distribution, making it both a synthetic and analytical measure.
\par Due to limitations in \citet{zenga1984proposta} definition like, making computation and interpretation less straightforward, \citet{zenga2007inequality} introduced a new inequality curve, denoted by $I(p)$ based on a simpler and more intuitive foundation. For each population proportion $p$, the curve compares the average income of the bottom $100 \times p\%$ (lower mean $M^{-}(p)$) of the population  with that of the top $100 \times (1-p)\%$ (upper mean $M^{+}(p)$), defined as 
\begin{equation}\label{1.2}
    I(p) =  1 - \frac{M^{-}(p)}{M^{+}(p)}.
\end{equation}
Unlike \eqref{1.1}, the modified Zenga curve \eqref{1.2} uniquely determines the underlying distribution, exhibits no forced behavior (unlike the Lorenz and Bonferroni curves), and yields a synthetic index 
$I = \int_{0}^{1} I(p)\,dp$, the total area beneath the curve over the unit interval that exceeds the Gini index in sensitivity, particularly at the extremes of the distribution.
\par The theoretical properties of the Zenga curve and index have been extensively studied in the literature. Various works have examined their fundamental properties, consistency with Lorenz ordering, relationship with the transfer principle, and connections with stochastic orders. In addition, several studies have focused on the estimation and inference of the Zenga index, including asymptotic results and empirical applications, as well as comparisons with classical measures such as the Gini index.  For example, see \citet{zenga1990concentration}, \citet{berti1995note}, \citet{pollastri1987},  \citet{polisicchio1993a,polisicchio1993b},  \citet{greselin2010zenga}, \citet{nair2012some}, \citet{langel2012inference}.
\par Despite these advances, the existing literature is largely confined to the univariate setting, where inequality is measured with respect to a single attribute, predominantly the income. However, socio-economic well being is inherently multidimensional, determined by the joint distribution of variables such as income, education, health, and access to infrastructure. In recent years, a large volume of empirical research has focused on analyzing such multivariate data across diverse fields, including economic development, finance, and information technology. In particular, several studies, such as \citet{chinn2007determinants}, \citet{billon2009disparities}, \citet{cruz2012digital}, and \citet{thoene2025broadband}, highlight the growing importance of digital inclusion, emphasizing that factors such as broadband penetration and digital literacy jointly determine access to opportunities in the digital economy.
\par The primary objective of the present study is to address this gap by proposing bivariate extensions of the Zenga curve using the bivariate quantile function introduced by \citet{vineshkumar2019bivariate, unnikrishnan2023properties} and subsequently extended to the multivariate case by \citet{pr2024multivariate}. We establish the theoretical properties of the proposed measures and develop nonparametric estimators for their practical implementation. The applicability of the bivariate Zenga curve is further illustrated through an empirical analysis of digital inequality, using data on broadband penetration and digital literacy across countries.
\par The paper is organized into six sections. Section 2 introduces the first moment distribution and derives the bivariate Zenga surface based on it, along with its properties. In Section 3, a new bivariate Zenga surface is proposed as an extension of the definition in \citet{zenga2007inequality}, and its properties are discussed in Section 4. Section 5 presents a vector-valued bivariate Zenga curve, examines its properties, develops corresponding estimators, and includes a simulation study. Finally, Section 6 concludes the paper with an empirical analysis of bivariate data using the vector-valued bivariate Zenga curve.

\section{First moment distribution and multivariate Zenga surface}
The extension of inequality measures to multiple dimensions requires a robust definition of a 
multivariate quantile function. As detailed by \citet{pr2024multivariate}, for an $n$-dimensional non-negative random vector $\mathbf{X} = (X_1, X_2, \dots, X_n)$ with an absolutely continuous joint distribution function $F(\mathbf{x})$ and survival function $\bar{F}(\mathbf{x})$ where $\mathbf{x} = (x_1, x_2, \dots, x_n)$, 
the multivariate quantile function is constructed recursively using conditional distributions. 
This approach defines the quantile vector $Q(\mathbf{u})$, with $\mathbf{u} = (u_1, u_2, \cdots, u_n)$, as
\begin{equation*}
    Q(\mathbf{u}) = \big(Q_1(u_1), \; Q_{21}(u_1,u_2), \; \dots, \; Q_{n(n - 1) \cdots 21}(u_1, u_2, \cdots, u_n)\big).
\end{equation*}
where
\begin{eqnarray*}
    Q_1 (u_1) &=& Q_1(u_1) = \inf \{x_1 \mid F_1(x_1) \geq u_1\}, \\
    Q_{21}(u_{21}) &=& Q_{21}(u_1, u_2) = \inf \{x_2 \mid F_{21}(Q_1, x_2) \geq u_{21}\}, \\
    \cdots && \cdots \\
    Q_{n\cdots1}(u_{n(n-1)\cdots1}) &=& Q_{n (n - 1) \cdots 21}(u_1, u_2, \cdots, u_n) \\
    &=& \inf \{x_n \mid F_{n(n - 1)\cdots1}(Q_1, Q_{21}, \cdots, Q_{(n-1)(n - 2) \cdots 21}, x_n) \geq u_{n\cdots1}\}
\end{eqnarray*}
where 
\begin{eqnarray*}
    F_{n(n - 1)\cdots1} \left(Q_1, \cdots, Q_{(n-1)(n - 2) \cdots 21}, x_n\right) = P\left(X_n \leq x_n \mid X_1 > Q_1(u_1), \cdots, X_{n - 1} > Q_{n - 1}(u_{n - 1})\right)
\end{eqnarray*}
denotes the $n$th order conditional distribution function conditioned on the preceding marginal quantiles.  This definition ensures that the multivariate quantile function comprehensively captures the dependence structure between the variables.  For various properties of $Q(\mathbf{u})$ we refer to \citet{unnikrishnan2023properties}.
\par To construct a multivariate Zenga surface, one has to define the multivariate first moment distribution.
\begin{defi}
For an $n$-dimensional non-negative random vector 
$\mathbf{X} = (X_1, X_2, \dots, X_n)$
with joint distribution function $F(\mathbf{x})$ and finite product moment 
$\mu = E\!\left(\prod_{i=1}^n X_i \right),$
the multivariate first moment distribution function $F^{(1)}(\mathbf{x})$ is defined as
\begin{equation}
    F^{(1)}(\mathbf{x}) = \frac{1}{\mu} 
\int_0^{x_1} \int_0^{x_2} \cdots \int_0^{x_n} 
t_1 t_2 \cdots t_n \, f(t_1,t_2,\dots,t_n)dt_1dt_2\ldots dt_n.
\end{equation}
where $f(t_1,t_2,\dots,t_n)$ is the joint probability density function.
\end{defi}
This satisfies the necessary boundary conditions  
$F^{(1)}(x_1,\dots,x_{i-1},0,x_{i+1},\dots,x_n) = 0$ for any  $i$, and
$\lim_{x_1\to \infty,\dots,x_n \to \infty} F^{(1)}(x_1,\ldots,x_n) = 1.$
Its corresponding density function is given by
\begin{equation}
   f^{(1)}(\mathbf{x}) = \frac{\prod_{i=1}^n x_i}{\mu} \, f(\mathbf{x}), 
\end{equation}
which explicitly weights the original density by the product of the variable values, assigning 
higher probability mass to regions where all $X_i$ are simultaneously large. The distribution 
$F^{(1)}(\mathbf{x})$ represents the proportion of the total product $X_1 X_2 \cdots X_n$ contributed by all 
components of $\mathbf{X}$ whose values are less than or equal to $\mathbf{x}$.

The quantile function of this multivariate first moment distribution, denoted $Q^{(1)}(\mathbf{u})$, is 
constructed by applying the conditional approach to $F^{(1)}(\mathbf{x})$. It is defined as the 
$n$-tuple
\begin{equation}\label{2.3}
   Q^{(1)}(u) = 
\big(Q_1^{(1)}(u_1), \; 
Q^{(1)}_{21}(u_1,u_2), \; 
\cdots, \; 
Q^{(1)}_{n\ldots1}(u_1,\dots,u_n)\big). 
\end{equation}
where
\begin{eqnarray*}
    Q_1^{(1)}(u_1) &=& \inf \{x_1 \mid F^{(1)}_1(x_1) \geq u_1\}, \\
    Q_{21}^{(1)}(u_1, u_2) &=& \inf \{x_2 \mid F^{(1)}_{21}(Q_1^{(1)}, x_2) \geq u_{21}\}, \\
    \cdots && \cdots \\
    Q_{n\ldots1}^{(1)}(u_1, u_2, \cdots, u_n) &=& \inf \{x_n \mid F^{(1)}_{n\ldots 1}(Q_1^{(1)}, Q_{21}^{(1)}, \cdots, Q^{(1)}_{n-1\cdots 1}, x_n) \geq u_{n\cdots1}\}.
\end{eqnarray*}

The following example illustrates the computation of $Q^{(1)} (u)$ given in the proposed formulation \eqref{2.3}.
\begin{ex}\label{zengaex1}
  Consider a bivariate Pareto distribution with survival function  \[
\bar{F}(x_1,x_2)=(1+x_1+x_2)^{-\alpha},\, x_1,x_2 \geq 0
\] 
and probability density function \[
f(x_1,x_2)=\alpha (\alpha+1)(1+x_1+x_2)^{-(\alpha+2)},
\]  
with $E(X_1X_2)=\frac{1}{(\alpha-1)(\alpha-2)}.$
Then the probability density function of the first moment distribution is obtained as 
\[
f^{(1)}(x_1,x_2) = x_1 x_2\alpha (\alpha+1)(1+x_1+x_2)^{-(\alpha+2)},\,x_1,x_2\geq0,\,\alpha>2
\]
and 
\[
\begin{aligned}
F^{(1)}(x_1,x_2) \;= \; & 1 - (1+x_1)^{-\alpha}\,\big(1+\alpha x_1+(\alpha-1)x_1^2\big)
- (1+x_2)^{-\alpha}\,\big(1+\alpha x_2+(\alpha-1)x_2^2\big) \\
&+ (1+x_1+x_2)^{-\alpha}\,\Big[1+\alpha x_1+\alpha x_2+(\alpha-1)x_1^2+(\alpha-1)x_2^2+\alpha(\alpha-1)x_1x_2\Big].
\end{aligned}
\]
The marginal first moment distribution 
\[
 F_1^{(1)}(x_1) = 1-(1+x_1)^{-\alpha}\,\big(1+\alpha x_1+(\alpha-1)x_1^2\big).
\]
\[
\begin{aligned}
   F^{(1)}_{21}(x_2\mid x_1)
= \; &\frac{1}{1-(1+x_1)^{-\alpha}\,\big(1+\alpha x_1+(\alpha-1)x_1^2\big)}
\big[1-(1+x_1)^{-\alpha}\big(1+\alpha x_1+(\alpha-1)x_1^2\big)\\
&-(1+x_2)^{-\alpha}\big(1+\alpha x_2+(\alpha-1)x_2^2\big)\\
&+(1+x_1+x_2)^{-\alpha}\big(1+\alpha x_1+\alpha x_2+(\alpha-1)x_1^2+(\alpha-1)x_2^2+\alpha(\alpha-1)x_1x_2\big)\big].
\end{aligned}
\]
The bivariate quantile function of the first moment distribution can be obtained by solving the equations $u_1=F^{(1)}_1(x_1)$ and $u_2=F^{(1)}_{21}(x_2|x_1)$. The quantile functions are implicitly defined and for high quantiles, the following asymptotic approximations hold:
\[
Q^{(1)}_1(u_1)\sim \left(\frac{\alpha-1}{1-u_1}\right)^{\frac{1}{\alpha-2}}
\] 
and 
\[
Q^{(1)}_{21}(u_1,u_2)\sim \left(\frac{\alpha-1}{1-u_1+u_1u_2}\right)^{\frac{1}{\alpha-2}}.
\]
Thus \begin{equation}\label{par_firstmoment}
  Q^{(1)}(\mathbf{u} )\sim \Bigg(\left(\frac{\alpha-1}{1-u_1}\right)^{\frac{1}{\alpha-2}}, \left(\frac{\alpha-1}{1-u_1+u_1u_2}\right)^{\frac{1}{\alpha-2}}\Bigg).
\end{equation}
\end{ex}
The multivariate Zenga surface is constructed as a direct generalization of its univariate and bivariate counterparts, providing a multidimensional surface that measures local inequality by comparing the quantiles of the original distribution to the quantiles of the first moment distribution across all variables.
\begin{defi}\label{def2.2}
For a vector of quantiles $\mathbf{u}=(u_1,u_2,\dots,u_n)$, the multivariate Zenga surface is defined as
\begin{equation}\label{definition}
   Z(\mathbf{u}) = Z(u_1,u_2,\dots,u_n) 
= 1 - \frac{Q_1(u_1)\, Q_{21}(u_1,u_2)\, \cdots\, Q_{n\dots1}(u_1,\dots,u_n)}
{Q_1^{(1)}(u_1)\, Q_{21}^{(1)}(u_1,u_2)\, \cdots\, Q_{n\dots1}^{(1)}(u_1,\dots,u_n)}. 
\end{equation}
\end{defi}
This surface measures the relative disparity at each point $\mathbf{u}$ between the product of the joint 
quantiles of the population distribution and the product of the joint quantiles of the share 
distribution of the total product $X_1 X_2 \cdots X_n$.
\par For $n=2$, the definition simplifies to the bivariate Zenga surface:
\[
Z(u_1,u_2) 
= 1 - \frac{Q_1(u_1)\, Q_{21}(u_1,u_2)}
{Q_1^{(1)}(u_1)\, Q_{21}^{(1)}(u_1,u_2)}.
\]
This represents a surface where the height at each point $(u_1,u_2)$ measures the inequality 
between the population and the income share for the respective proportions of the population.
For $n=1$, the definition reduces to the classical univariate Zenga curve \eqref{1.1}:
\[
Z(u) = 1 - \frac{Q(u)}{Q^{(1)}(u)},
\]
which compares the quantiles of the population distribution with those of the income share 
distribution.
This multivariate formulation provides a comprehensive tool for analyzing the complex, 
interconnected structure of inequality between multiple variables, generalizing the intuitive 
comparison of distribution and first-moment quantiles to higher dimensions.
\begin{ex}
    Consider the bivariate Pareto distribution given in Example \ref{zengaex1}, with marginal distribution function, $$F_1(x)=1-(1+ x_1)^{-\alpha}$$ 
    and the conditional distribution $$F_{21}(x_2|x_1)= 1-\left(\frac{1+x_1+x_2}{1+x_1}\right)^{-\alpha}.$$
    The bivariate quantile function for the bivariate Pareto distribution can be obtained  by solving $u_1=F_1(x_1)$ and $u_2=F_{21}(x_2|x_1)$. Thus $$Q_1(u_1)=(1-u_1)^{-\frac{1}{\alpha}}-1$$
    and $$Q_{21}(u_1,u_2)=(1-u_1)^{-\frac{1}{\alpha}}\big((1-u_2)^{-\frac{1}{\alpha}}-1\big).$$
    Hence \begin{equation}\label{zenga_paretoqauntile}
    Q(\mathbf{u})=\big( (1-u_1)^{-\frac{1}{\alpha}}-1  ,  (1-u_1)^{-\frac{1}{\alpha}}\big((1-u_2)^{-\frac{1}{\alpha}}-1\big)\big) .   
    \end{equation}
    The bivariate asymptotic quantile function for the first moment distribution is given in \eqref{par_firstmoment}.
    Thus, the bivariate Zenga surface for the bivariate Pareto distribution can be obtained from \eqref{definition} as 
    \begin{equation}\label{zenga_pareto}
    \begin{aligned}
        Z(\mathbf{u})=&Z(u_1,u_2)\sim 1-\frac{ ((1-u_1)^{-\frac{1}{\alpha}}-1)   (1-u_1)^{-\frac{1}{\alpha}}\big((1-u_2)^{-\frac{1}{\alpha}}-1\big)}{(\frac{\alpha-1}{1-u_1})^{\frac{1}{\alpha-2}}(\frac{\alpha-1}{1-u_1+u_1u_2})^{\frac{1}{\alpha-2}}}\\
      \end{aligned}  
    \end{equation}
    The graph plotted for $\alpha=3$ is given in Figure \ref{paretozengacurve}.
    \begin{figure}
        \centering
        \includegraphics[width=0.5\linewidth]{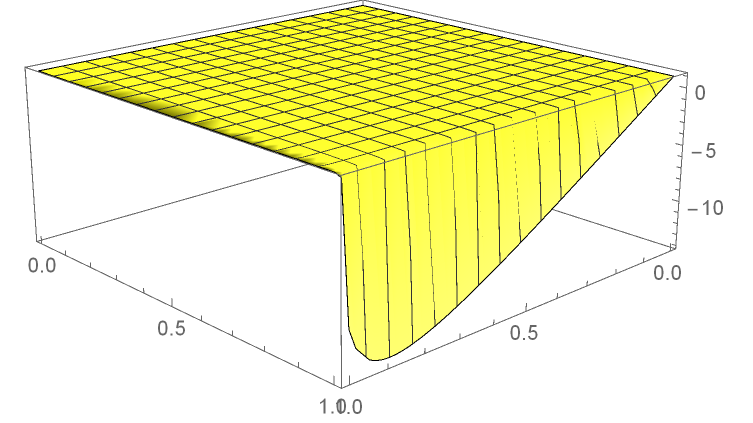}
        \caption{Bivariate Zenga surface for Pareto distribution}
        \label{paretozengacurve}
    \end{figure}
\end{ex}

The properties of the multivariate Zenga surface can be most conveniently derived and illustrated 
first in the bivariate case, with the understanding that the logic extends naturally to higher 
dimensions. This approach is standard, as the bivariate case captures the essential structure 
involving conditional distributions and dependence, without the overwhelming notation of the 
general multivariate setting.

In the univariate case, it is established that the Zenga curve $Z(u)$ is bounded between $0$ and $1$. 
We now examine whether Definition \ref{def2.2} possess this property by considering the bivariate case, as the multivariate extension is straightforward.

\begin{thm}[Boundedness of the Bivariate Zenga surface]\label{thm2.1}
\par Let $(X_1, X_2)$ be a pair of non-negative random variables with a finite mean product 
$\mu = E(X_1 X_2) < \infty$. 
The bivariate Zenga surface $Z(u_1,u_2)$, defined for $0 < u_1, u_2 < 1$, is given by
\[
Z(u_1,u_2) 
= 1 - \frac{Q_1(u_1)\, Q_{21}(u_1,u_2)}
{Q_1^{(1)}(u_1)\, Q_{21}^{(1)}(u_1,u_2)},
\]
and is bounded such that
\[
0 \leq Z(u_1,u_2) < 1.
\]
\end{thm}
\begin{proof}
We prove the two inequalities separately.  To show that:
\[
Z(u_1,u_2)\geq 0\iff 1-\frac{Q_1(u_1)Q_{21}(u_1,u_2)}{Q_1^{(1)}(u_1)Q_{21}^{(1)}(u_1,u_2)}\geq 0 \iff 
\frac{Q_1(u_1)Q_{21}(u_1,u_2)}{Q_1^{(1)}(u_1)Q_{21}^{(1)}(u_1,u_2)}\leq 1.
\]
This requires showing that
\[
Q_1(u_1)\leq Q_1^{(1)}(u_1),\qquad Q_{21}(u_1,u_2)\leq Q_{21}^{(1)}(u_1,u_2).
\]
The first-moment distribution $F^{(1)}(x_1,x_2)$ is defined as
\[
dF^{(1)}(x_1,x_2)=\frac{x_1x_2}{\mu}\,dF(x_1,x_2).
\]
For $x_1,x_2>0$, a non-decreasing weight function
$w(x_1, x_2) = x_1 x_2,$ it follows from  \cite{jain1995multivariate}
that the weighted version is stochastically larger than the original distribution. Hence, 
\[
F^{(1)}(x_1,x_2)\leq F(x_1,x_2).
\]
Since $F^{(1)}(x_1)\leq F_1(x_1)$ for all $x_1$, it follows that 
\[
\{x_1 : F^{(1)}(x_1)\geq u_1\} \subseteq \{x_1 : F_1(x_1)\ge u_1\}.
\]
Taking infimum on both sides yields, $\inf\{x_1 : F^{(1)}(x_1)\ge u_1\} \ge \inf\{x_1 : F_1(x_1)\ge u_1\}$, or equivalently,
\[
Q_1^{(1)}(u_1) \geq Q_1(u_1),
\]
for all $0<u_1<1$.
Further, for each fixed $x_1$, the conditional density of the weighted distribution is
\[
f^{(1)}_{21}(x_2 \mid x_1)
= \frac{x_2 f_{21}(x_2 \mid x_1)}{E[X_2 \mid X_1=x_1]},
\]
which is the weighted version of $X_2 \mid X_1=x_1$. Therefore, $F^{(1)}_{21}(x_2 \mid x_1)\le F_{21}(x_2 \mid x_1)$, and further $Q_{21}^{(1)}(u_1,u_2)\ge Q_{21}(u_1,u_2)$.  Thus, we obtain
\[
Q_1^{(1)}(u_1)Q_{21}^{(1)}(u_1,u_2)\geq Q_1(u_1)Q_{21}(u_1,u_2),
\]
proves, $Z(u_1,u_2)\geq 0$.
\par To show that:
\[
Z(u_1,u_2)<1\iff 1-\frac{Q_1(u_1)Q_{21}(u_1,u_2)}{Q_1^{(1)}(u_1)Q_{21}^{(1)}(u_1,u_2)}<1\iff 
\frac{Q_1(u_1)Q_{21}(u_1,u_2)}{Q_1^{(1)}(u_1)Q_{21}^{(1)}(u_1,u_2)}>0.
\]
Since $X_1,X_2>0$, all quantiles $Q_1(u_1),Q_{21}(u_1,u_2),Q_1^{(1)}(u_1),Q_{21}^{(1)}(u_1,u_2)$ are positive. Hence, the fraction is positive, which proves that $Z(u_1,u_2)<1$.
Combining both parts, we conclude that 
$0 \leq Z(u_1,u_2) < 1$ for all  $0<u_1,u_2<1$.
\end{proof}
\begin{thm}[Perfect Equality for the Bivariate Zenga surface]\label{thm2.2}  
Let $(X_1, X_2)$ be a pair of non-negative random variables with a finite mean product $\mu = E(X_1 X_2) < \infty$.  
The bivariate Zenga surface $Z(u_1,u_2)$ is identically zero for all $u_1,u_2 \in (0,1)$, if and only if the distribution is degenerate,  
i.e., there exist constants $c_1>0$ and $c_2>0$ such that $X_1=c_1$ and $X_2=c_2$ almost surely.  
This represents a state of perfect equality where $X_1$ and $X_2$ are equally allocated in direct proportion to the population.  
\end{thm}
\begin{proof}
We prove both directions of the implication. First, assume $X_1=c_1$ and $X_2=c_2$ almost surely, with $c_1,c_2>0$.  

In this case of a degenerate distribution, the original quantile functions are simply the constants themselves,  
$Q_1(u_1)=c_1$ and $Q_{21}(u_1,u_2)=c_2$ for all $u_1,u_2$.  

The mean $\mu=c_1c_2$.  
The first moment distribution is defined by  
\[
dF^{(1)}(x_1,x_2)=\frac{x_1x_2}{\mu}\, dF(x_1,x_2).
\]  
Since the original distribution $dF(x_1,x_2)$ is a point mass at $(c_1,c_2)$, the first moment density becomes  
\[
dF^{(1)}(x_1,x_2) = \frac{c_1c_2}{c_1c_2} \, dF(x_1,x_2) = dF(x_1,x_2).
\]  
Consequently, the first moment distribution is identical to the original distribution, and its quantile functions are also the constants. Thus $Q_1^{(1)}(u_1)=c_1$ and $Q_{21}^{(1)}(u_1,u_2)=c_2$ for all $u_1,u_2$.  Substituting these into the definition of the Zenga surface yields  
\[
Z(u_1,u_2) = 1 - \frac{c_1 \cdot c_2}{c_1 \cdot c_2} = 0
\]  
for all $u_1,u_2$. Thus, perfect equality implies the Zenga surface is identically zero.  

Conversely, assume $Z(u_1,u_2)=0$ for all $u_1,u_2 \in (0,1)$.  
This assumption implies that  
\[
1 - \frac{Q_1(u_1)\cdot Q_{21}(u_1,u_2)}{Q_1^{(1)}(u_1)\cdot Q_{21}^{(1)}(u_1,u_2)}=0,
\]  
which is equivalent  to $Q_1(u_1)\cdot Q_{21}(u_1,u_2)=Q_1^{(1)}(u_1)\cdot Q_{21}^{(1)}(u_1,u_2)$.  

From the property of stochastic dominance, we know that $Q_1^{(1)}(u_1)\geq Q_1(u_1)$ and $Q_{21}^{(1)}(u_1,u_2)\geq Q_{21}(u_1,u_2)$.  
For the product of these larger quantiles to equal the product of the smaller ones, and all the terms are positive the individual equalities must hold for all $u_1,u_2$;  
that is, $Q_1^{(1)}(u_1)=Q_1(u_1)$ and $Q_{21}^{(1)}(u_1,u_2)=Q_{21}(u_1,u_2)$.  

The equality of the first marginal quantile functions, $Q_1(u_1)=Q_1^{(1)}(u_1)$ for all $u_1$,  
implies that the first marginal distribution of the first moment distribution is identical to the first marginal of the original distribution.  
This identity, $F^{(1)}_1(x_1)=F_1(x_1)$, holds if and only if $X_1=c_1>0$ almost surely.  

Similarly, given $X_1=c_1$, the equality of the conditional quantiles $Q_{21}(u_1,u_2)=Q_{21}^{(1)}(u_1,u_2)$ for all $u_2$  
implies that the conditional distribution of $X_2$ given $X_1=c_1$ is the same under both $F$ and $F^{(1)}$.  
This identity holds if and only if $X_2=c_2>0$ almost surely.  

Therefore, $Z(u_1,u_2)=0$ for all $u_1,u_2$ if and only if $X_1=c_1$ and $X_2=c_2$ almost surely,  
which is the state of perfect equality.  
\end{proof}
In the univariate case, the Zenga index is obtained as an average of the Zenga curve over the unit interval. 
Motivated by this, we extend the idea to the bivariate setting by it as the average of the bivariate Zenga surface over the unit square.
\begin{defi}
Let $Z(u_1,u_2)$ be the bivariate Zenga surface defined on $[0,1]^2$. 
The corresponding bivariate Zenga index is defined as
\[
\xi = \int_0^1 \int_0^1 Z(u_1,u_2)\, du_1\, du_2.
\]
\end{defi}

\section{A new bivariate Zenga Surface}
Another version of the Zenga curve was introduced by \citet{zenga2007inequality} based on the conditional expectation that uniquely determines the distribution.  Let $X$ be a non-negative continuous random variable defined on
$0 \le x \le b < \infty$, with distribution function $F(x)$, density function $f(x)$,
and finite positive mean $\mu$.  The measure of inequality is defined by \citet{zenga2007inequality},
\begin{equation}\label{new zenga}
    A(x) = 1 - \frac{\mu^{-}(x)}{\mu^{+}(x)},
\end{equation}
where $\mu^{-}(x)=E(X|X\leq x)$ and $\mu^{+}(x)=E(X|X> x)$ denote the lower and upper partial means, respectively. Unlike $Z(u)$, it has been shown by \citet{zenga2007inequality} that $A(\cdot)$ in \eqref{new zenga} uniquely determines the underlying distribution $F(\cdot)$.  Motivated by this, we propose a bivariate extension of the Zenga curve \eqref{new zenga} that uniquely determines the distribution.
\par Let $(X_1,X_2)$ be a bivariate  non-negative absolutely continuous random vector with joint distribution function $F(x_1,x_2)$, probability density function $f(x_1,x_2)$, with finite expectation $E(X_1 X_2)$.  Then a bivariate measure of inequality $A(x_1,x_2)$, as an extension of the Zenga curve \eqref{new zenga} is defined as,
\begin{equation}\label{newzengabiv}
   A(x_1, x_2) = 1 - \frac{\mu^{-}(x_1, x_2)}{\mu^{+}(x_1, x_2)}, 
\end{equation}
where
\begin{equation}\label{3.3}
    \mu^{-}(x_1, x_2)
= \mathbb{E}(X_1 X_2 \mid X_1 \le x_1,\; X_2 \le x_2) = \frac{1}{F(x_1, x_2)}\int_{0}^{x_1} \int_{0}^{x_2} t_1 t_2 f(t_1, t_2)\, dt_1\, dt_2,
\end{equation}
and
\begin{equation}\label{3.4}
    \mu^{+}(x_1, x_2)
= \mathbb{E}(X_1 X_2 \mid X_1 > x_1,\; X_2 > x_2) = \frac{1}{\bar{F}(x_1, x_2)} \int_{x_1}^{\infty} \int_{x_2}^{\infty} t_1 t_2 f(t_1, t_2)\, dt_1\, dt_2,
\end{equation}
denote respectively the lower and upper partial product moments in the bivariate case.  The measure compares the average and captures the relative standing of the
lower income group compared to the upper income group, and provides a natural inequality measure that ranges between $0$ and $1$.  The multivariate analogue of \eqref{newzengabiv} is straightforward.
\begin{ex}\label{ex3.1}
     Consider a bivariate Pareto distribution with joint survival function
\[
\bar F(x_1,x_2)
=
(x_1+x_2-1)^{-c},
\qquad 0<x_1, x_2<1; \, x_1+x_2\geq 1; \, c>0,
\]
and joint probability density function $f(x_1,x_2)
=
c(c+1)\,(x_1+x_2-1)^{-c-2}, \;
0<x_1, x_2<1$ 
such that the partial moments, 
\[
\mu^-(x_1,x_2) = {\frac{(x_2-1)^{1-c} - (x_1+x_2-1)^{1-c}}{(c+1)\left[(x_1+x_2-1)^{-c} - x_1^{-c} - x_2^{-c} + 1\right]}}
\]
and
\[
\mu^{+}(x_1, x_2) = \frac{x_1^{1-c} - (x_1+x_2-1)^{1-c}}{c+1}
\]
Then, 
\[
A(x_1,x_2) = 1-\frac{\left((x_2-1)^{1-c} - (x_1+x_2-1)^{1-c}\right)\left(x_1^{1-c} - (x_1+x_2-1)^{1-c}\right)}{x_1^{1-c} - (x_1+x_2-1)^{1-c}}.
\]

\end{ex}
To handle various types of random variables, including discrete, continuous and mixed types
in higher dimension and to address the distribution functions that lack closed form expressions and the distributions that are specified by quantile functions we need to obtain the bivariate inequality $A(x_1,x_2)$ in terms of quantile functions.
From \citet{unnikrishnan2023properties}, the bivariate quantile function is given as a vector, $Q(u_1,u_2) = \left(Q_1(u_1), Q_{21}(u_1, u_2)\right)$. We have the transformations $u_1=F_1(x_1)$, $u_2=F_{21}(Q_1,Q_{21})$, $F(Q_1 (u_1), Q_{21}(u_1, u_2)) = u_1u_{2}$ and $\bar F(Q_1 (u_1), Q_{21}(u_1, u_2)) = (1-u_1)(1-u_{2})$. Also from \citet{unnikrishnan2023properties}, for any measurable function $g$,
\[
\int_{\mathbb{R}^2} g(x_1, x_2)F_1(dx_1) F_{21}(dx_2)
=
\int_{I^2} g(Q(u_1,u_2)) du_1 du_2.
\]
Applying these, the lower and upper product partial moments \eqref{3.3} and \eqref{3.4} in terms of the bivariate quantile framework can be written respectively as
\[
M^{-}(u_1, u_2)
=
\frac{1}{u_1 u_2}
\int_{0}^{u_1} \int_{0}^{u_2}
Q_1(p_1)\, Q_{21}(p_1,p_2) dp_1 dp_2,
\]
and
\[
M^{+}(u_1, u_2)
=
\frac{1}{(1-u_1)(1-u_2)}
\int_{u_1}^{1} \int_{u_2}^{1}
Q_1(p_1)\, Q_{21}(p_1,p_2) dp_1 dp_2.
\]
Then, the quantile based bivariate Zenga measure corresponding to $A(x_1, x_2)$ is given by
\begin{equation}\label{newbiv_zenga_quantile}
    I(u_1, u_2) = 1 - \frac{M^{-}(u_1, u_2)}{M^{+}(u_1, u_2)}, \; \text{for all } \; u_1,u_2 \in (0,1).
\end{equation}
\par To extend the measure to the multivariate case, we proceed as follows.  Let $(X_1, X_2, \ldots, X_n)$ be a random vector with an absolutely continuous  joint distribution $F(x_1, x_2, \ldots, x_n)$ and $\mathbb{E}(X_1 X_2 \cdots X_n) < \infty$.  Then the multivariate Zenga measure is defined as
\begin{small}
 \begin{equation}\label{zenga_multi}
    I(u_1, u_2, \ldots, u_n)
=
1 -
\frac{
\displaystyle
\left(\prod_{i=1}^{n} (1-u_i)\right)
\int_{0}^{u_1} \int_{0}^{u_2} \cdots \int_{0}^{u_n}
 Q_1(p_1)Q_{21}(p_1,p_2)\cdots Q_{n\cdots1}(p_1, p_2, \ldots, p_n)\, dp_1 \cdots dp_n
}{
\displaystyle
\left(\prod_{i=1}^{n} u_i\right) 
\int_{u_1}^{1} \int_{u_2}^{1} \cdots \int_{u_n}^{1}
 Q_1(p_1)Q_{21}(p_1,p_2)\cdots Q_{n\cdots1}(p_1, p_2,\ldots, p_n)\, dp_1 \cdots dp_n
}. 
 \end{equation}   
\end{small}
For $n = 1$, this expression reduces to the 
\[
I(u)=1 -\frac{(1-u) \int_{0}^{u} Q(p)\, dp}{u \int_{u}^{1} Q(p)\, dp},
\]
which coincides with the univariate  definition due to \citet{nair2012some} based on the Zenga curve given in \eqref{new zenga}.
\begin{ex}
Consider the distribution defined by quantile functions
\[
Q_1(u_1) = K_1 u_1^{b_1}, 
\qquad
Q_{21}(u_1, u_2) = K_2 u_1^{b_1} u_2^{b_2},
\]
where $K_1,K_2>0; \; 0 \leq u_1, u_2 \leq 1$ and $b_1,b_2>0$.  Then
\[
x_1 = K_1 u_1^{b_1}
\quad \Rightarrow \quad
u_1 = \left(\frac{x_1}{K_1}\right)^{1/b_1},
\]
and
\[
x_2 =  K_2 u_1^{b_1} u_2^{b_2}
\quad \Rightarrow \quad
u_2 = \left(\frac{K_1x_2}{K_2x_1 }\right)^{1/b_2}.
\]
Hence, the joint distribution function can be written as
\[
F(x_1,x_2)=\left(\frac{x_1}{K_1}\right)^{1/b_1}
\left(\frac{K_1x_2}{K_2x_1 }\right)^{1/b_2}.
\]
For $0<x_1,x_2<1$ this is  a bivariate power distribution.  Then the Zenga surface \eqref{newbiv_zenga_quantile} becomes,
\[
I(u_1,u_2) = \frac{1 - u_1^{b_1+1} - u_2^{b_2+1}}{(u_1^{b_1+1}-1)(u_2^{b_2+1}-1)}.
\]
The plot for bivariate Zenga surface of the power distribution for $b_1 = 2, b_2 = 3$ is given in Figure \ref{zenga_power}.
\begin{figure}
        \centering
     \includegraphics[width=0.5\linewidth]{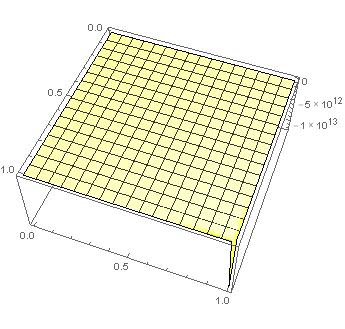}
        \caption{Bivariate Zenga Surface for Power distribution}
        \label{zenga_power}
    \end{figure}
\end{ex}
\section{Properties of $I(u_1,u_2)$}

We now examine the properties of the bivariate Zenga surface $I(u_1,u_2)$.  The first property is based on the RTI, which is defined as follows.
\begin{defi}
   Let $(X_1, X_2)$ denote a continuous random vector with joint distribution function $F(\cdot, \cdot)$.  Then $X_i$ is said to be right tail increasing (RTI) in $X_j$ if
   ${P(X_i>x_i|X_j>x_j)}$ is increasing in $x_i$ for all $x_j$,
   or equivalently
   $P(X_i\leq x_i|X_j>x_j)$ is decreasing in $x_i$ for all $x_j$ where $i, j = 1,2, \; i \neq j$.
   
\end{defi}
\begin{thm}\label{thm4.1}
 Under RTI, the bivariate Zenga surface bounded by, $0 \le I(u_1,u_2) \le 1.$
\end{thm}
\begin{proof}
  Since the quantile functions are positive and the integrals are taken over
positive regions, we have
\[
M^{-}(u_1,u_2) > 0, \qquad M^{+}(u_1,u_2) > 0 .
\]
Further, $Q_1(u_1)$ is non-decreasing in $u_1$ by the property of quantile functions. Under RTI, $P(X_1 \leq x_1|X_2 > x_2)$ is decreasing in $x_1$ for all $x_2$.  This means that as the conditioning value $X_1$ increases, the conditional distribution shifts to the right. Consequently, the conditional quantile function $Q_{21}(u_1, u_2)$ is non-decreasing in $u_1$. It is also non-decreasing in $u_2$ by the properties of quantile function. Hence the product $Q_1(u_1)Q_{21}(u_1, u_2)$ is non-decreasing in each argument.  Let
\[
A = \int_{0}^{u_1} \int_{0}^{u_2} Q_1(p_1)Q_{21}(p_1,p_2)\, dp_1\, dp_2, \quad \text{and} \quad B = \int_{u_1}^{1} \int_{u_2}^{1} Q_1(p_1)Q_{21}(p_1,p_2)\, dp_1\, dp_2.
\]

We need to show that $\frac{A}{u_1u_2} \le \frac{B}{(1-u_1)(1-u_2)}$.  By the monotonicity of $Q_1$ and $Q_{21}$, for
$(p_1,p_2)\in [0,u_1]\times[0,u_2]$ and
$(r_1,r_2)\in [u_1,1]\times[u_2,1]$, we have
\[
Q_1(p_1)Q_{21}(p_1,p_2) \le Q_1(r_1)Q_{21}(r_1,r_2).
\]
Thus,
\[
\int_{0}^{u_1}\!\!\int_{0}^{u_2}
\int_{u_1}^{1}\!\!\int_{u_2}^{1}
\bigl[ Q_1(p_1)Q_{21}(p_1,p_2) - Q_1(r_1)Q_{21}(r_1,r_2) \bigr]
\, dp_1 dp_2 dr_1 dr_2 \le 0.
\]
This implies
\[
 (1-u_1)(1-u_2) A \le u_1u_2B,
\]
and hence
\[
0\leq M^{-}(u_1,u_2) \le M^{+}(u_1,u_2),
\]
which completes the proof of the theorem.
\end{proof}

\begin{thm}
  Let $(X_1,X_2)$ be non-negative  with $\mu = \mathbb{E}(X_1X_2) < \infty$ and satisfies the RTI property.  Then $I(u_1,u_2)=0 \quad \text{for all} \; u_1, u_2 \in (0,1)$ if and only if $(X_1,X_2)$ is degenerate. 
\end{thm}

\begin{proof}
First, assume that $(X_1,X_2)$ is degenerate, then $X_1=c_1, X_2=c_2. \;a.s$ and correspondingly,
\[
M^{-}(u_1,u_2)=c_1c_2,\qquad
M^{+}(u_1,u_2)=c_1c_2,
\implies
I(u_1,u_2)=0.
\]

Conversely, assume that $I(u_1,u_2)=0$ for all $u_1,u_2 \in (0,1)$. Then
\[
M^{-}(u_1,u_2)=M^{+}(u_1,u_2)\quad \forall u_1,u_2.
\]
Hence
\begin{equation}\label{rtithm}
 (1-u_1)(1-u_2)\,A = u_1 u_2\,B \quad \forall u_1,u_2,   
\end{equation}
where
\[
A=\int_0^{u_1}\int_0^{u_2} Q_1(p_1)\,Q_{21}(p_1,p_2)\,dp_2\,dp_1,
\quad
B=\int_{u_1}^{1}\int_{u_2}^{1} Q_1(p_1)\,Q_{21}(p_1,p_2)\,dp_2\,dp_1.
\]
Let
\[
g(p_1,p_2)=Q_1(p_1)\,Q_{21}(p_1,p_2).
\]

Under the RTI assumption, $Q_{21}(p_1,p_2)$ is non-decreasing in both arguments and $Q_1(p_1)$ is non-decreasing in $p_1$. Hence $g$ is non-decreasing in each coordinate, so
$g(p_1,p_2)\le g(r_1,r_2)$
whenever  $p_1\le r_1,\; p_2\le r_2$.
Fix $u_1,u_2\in(0,1)$. Then for 
$(p_1,p_2)\in [0,u_1]\times[0,u_2]$ and
$(r_1,r_2)\in [u_1,1]\times[u_2,1]$, we have
\[
g(p_1,p_2)-g(r_1,r_2)\le 0.
\]

Integrating, we obtain
\[
\int_{0}^{u_1}\int_{0}^{u_2}
\int_{u_1}^{1}\int_{u_2}^{1}
\bigl[ g(p_1,p_2) - g(r_1,r_2) \bigr]
\, dr_2\,dr_1\,dp_2\,dp_1 \le 0.
\]

Thus
\[
(1-u_1)(1-u_2)\,A - u_1 u_2\,B \le 0
\]
From \eqref{rtithm} we have, \[
(1-u_1)(1-u_2)\,A -u_1 u_2\,B = 0
\]
Since the integrand is non-positive and its integral is zero, it follows that
\[
g(p_1,p_2)-g(r_1,r_2)=0 \quad \text{a.e.}
\]
for almost all $(p_1,p_2)\in[0,u_1]\times[0,u_2]$ and 
$(r_1,r_2)\in[u_1,1]\times[u_2,1]$.

This implies that $g$ is constant almost everywhere. Hence there exists $C\ge 0$ such that
\[
Q_1(p_1)\,Q_{21}(p_1,p_2)=C \quad \text{a.e.}
\]

Fix $p_1$. Since $Q_{21}(p_1,p_2)$ is non-decreasing in $p_2$ and the product is constant in $p_2$, it follows that $Q_{21}(p_1,p_2)$ is independent of $p_2$. Hence
\[
Q_{21}(p_1,p_2)=\beta(p_1)
\]
for some function $\beta$. Thus
\[Q_1(p_1)\,\beta(p_1)=C.\]

Under RTI, $\beta(p_1)$ is non-decreasing, and $Q_1(p_1)$ is also non-decreasing. If $Q_1$ were not constant, there would exist $p_1<p_1'$ such that $Q_1(p_1)<Q_1(p_1')$, which implies
\[
\beta(p_1)=\frac{C}{Q_1(p_1)} > \frac{C}{Q_1(p_1')}=\beta(p_1'),
\]
contradicting monotonicity. Hence $Q_1(p_1)=c_1$ is constant, and then $\beta(p_1)=c_2$ is also constant. Therefore
\[
X_1=c_1,\quad X_2=c_2 \quad \text{a.s.}
\]
Hence $(X_1, X_2)$ is degenerate.
\end{proof}

\begin{thm}[Relationship between the Zenga Surface and the Lorenz surface]
Let $(X_1,X_2)$ be a non-negative absolutely continuous random vector with
$\mu = \mathbb{E}(X_1X_2) < \infty$.
Let $L(u_1,u_2)$ denote the bivariate Lorenz surface defined by
\[
L(u_1,u_2)
=
\frac{1}{\mu}
\int_{0}^{u_1}\int_{0}^{u_2}
Q_1(p_1)Q_{21}(p_1,p_2)\, dp_1\, dp_2
\] (\citet{shifna2025extending}).  Then the bivariate Zenga surface satisfies
\[
I(u_1,u_2)
=
1 -
\frac{(1-u_1)(1-u_2)\,L(u_1,u_2)}
{u_1u_2\,[1 - L(1,u_2) - L(u_1,1) + L(u_1,u_2)]},
\quad (u_1,u_2)\in(0,1)^2.
\]
\end{thm}

\begin{proof} We have $L(u_1,u_2)
=
\frac{1}{\mu}
\int_{0}^{u_1}\int_{0}^{u_2}
Q_1(p_1)Q_{21}(p_1,p_2)\, dp_1\, dp_2$.  Then
\[
M^{-}(u_1,u_2)
=
\frac{\mu\, L(u_1,u_2)}{u_1u_2} \quad \text{and} \quad M^{+}(u_1,u_2)
=
\frac{\mu\,[1 - L(1,u_2) - L(u_1,1) + L(u_1,u_2)]}{(1-u_1)(1-u_2)}.
\]
Hence,
\[
I(u_1,u_2)
=
1 -
\frac{(1-u_1)(1-u_2)\,L(u_1,u_2)}
{u_1u_2\,[1 - L(1,u_2) - L(u_1,1) + L(u_1,u_2)]}.
\] 
\end{proof}
In the univariate case, a synthetic index was associated with the Zenga curve and we extend this to the bivariate case as well.
\begin{defi}
Let $I(u_1,u_2)$ be the new bivariate Zenga surface defined on $[0,1]^2$. 
The associated bivariate synthetic index is defined as
\[
I = \int_0^1 \int_0^1 I(u_1,u_2)\, du_1\, du_2.
\]
\end{defi}

\section{Vector-valued bivariate  Zenga curve}
Let $(X_1,X_2)$ be a non-negative absolutely continuous random vector with finite positive means.
We represent the joint distribution through the bivariate quantile functions
\[
(Q_{12}(u_1, u_2), Q_2(u_2))
\quad \text{or equivalently} \quad
(Q_1(u_1), Q_{21}(u_1, u_2)).
\]
Now, define
\[
J_{12}(u_1,u_2) = \int_0^{u_1} Q_{12}(p, u_2)\,dp,
\qquad
J_{21}(u_1,u_2) = \int_0^{u_2} Q_{21}(u_1, p)\,dp.
\]
Also, let the conditional means in terms of the quantile functions by,
\[
\mu_{12}(u_2) = \int_0^1 Q_{12}(p, u_2)\,dp = E(X_1 \mid X_2 ),
\]
and
\[
\mu_{21}(u_1) = \int_0^1 Q_{21}(u_1, p)\,dp = E(X_2 \mid X_1 ).
\]
Then the vector-valued bivariate Zenga curve (VBZC) is defined as
\[
\mathbf{I}(u_1,u_2) = \bigl(I_{12}(u_1,u_2), I_{21}(u_1,u_2)\bigr),
\]
where
\[
I_{12}(u_1,u_2)
=
1 - \frac{M^-_{12}(u_1,u_2)}{M^+_{12}(u_1,u_2)},
\quad
I_{21}(u_1,u_2)
=
1 - \frac{M^-_{21}(u_1,u_2)}{M^+_{21}(u_1,u_2)},
\]
with $M^-_{12}(u_1,u_2)$ and $M^+_{12}(u_1,u_2)$ denote respectively by
\[
M^-_{12}(u_1,u_2) = \frac{1}{u_1} \int_{0}^{u_1}{Q_{12}(p, u_2)dp} = \frac{J_{12}(u_1,u_2)}{u_1} ,
\]
\[
M^+_{12}(u_1,u_2) = \frac{1}{1-u_1}
\int_{u_1}^1 Q_{12}(p, u_2)\,dp = \frac{\mu_{12}(u_2) - J_{12}(u_1, u_2)}{1 - u_1},
\]
and analogous definitions for $M^-_{21}$ and $M^+_{21}$.

\begin{ex}
    Consider the bivariate Pareto distribution with joint survival function given in Example \ref{ex3.1},
\[
\bar F(x_1,x_2) = (x_1+x_2-1)^{-c}, \qquad 0<x_1, x_2<1; x_1+x_2\geq 1; c>0.
\]
with joint distribution function
\[
F(x_1,x_2)
=
(x_1+x_2-1)^{-c}-x_1^{-c}-x_2^{-c}+1.
\]
such that the conditional quantile functions are given by
\[
Q_{12}(u_1, u_2)
=
1+(1-u_2)^{-1/c}\Bigl[(1-u_1)^{-1/c}-1\Bigr],
\]
\[
Q_{21}(u_1, u_2)
=
1+(1-u_1)^{-1/c}\Bigl[(1-u_2)^{-1/c}-1\Bigr].
\]
Then,
\begin{equation}\label{5.1}
I_{12}(u_1,u_2)  =1-\frac{(1-u_1)\left(\left(1 - u_{1}\right)^{\frac{1}{c}} \left(\left(c - 1\right) \left(1 - u_{2}\right)^{\frac{1}{c}} u_{1} + \left(1 - c\right) u_{1} + c\right) + cu_{1} - c\right)}{u_1\left(\left(\left(c - 1\right) \left(1 - u_{2}\right)^{\frac{1}{c}} - c + 1\right) \left(1 - u_{1}\right)^{\frac{1}{c}} + c\right) \left(u_{1} - 1\right)},
\end{equation}
and
\begin{equation}\label{5.2}
    I_{21}(u_1,u_2)=1-\frac{(1-u_2)\left(\left(1 - u_{2}\right)^{\frac{1}{c}} \left(\left(c - 1\right) \left(1 - u_{1}\right)^{\frac{1}{c}} u_{2} + \left(1 - c\right) u_{2} + c\right) + cu_{2} - c\right)}{u_2\left(\left(\left(c - 1\right) \left(1 - u_{1}\right)^{\frac{1}{c}} - c + 1\right) \left(1 - u_{2}\right)^{\frac{1}{c}} + c\right) \left(u_{2} - 1\right)}.
\end{equation}
Then VBZC is given by,
\begin{equation*}
    \mathbf{I}(u_1, u_2) = \left(I_{12}(u_1, u_2), I_{21}(u_1, u_2)\right).
\end{equation*}   
The plots \eqref{5.1} and \eqref{5.2} for $c=2$ is shown in Figure \ref{i12pareto} and \ref{i21pareto}.
\begin{figure}
    \centering
    \includegraphics[width=0.5\linewidth]{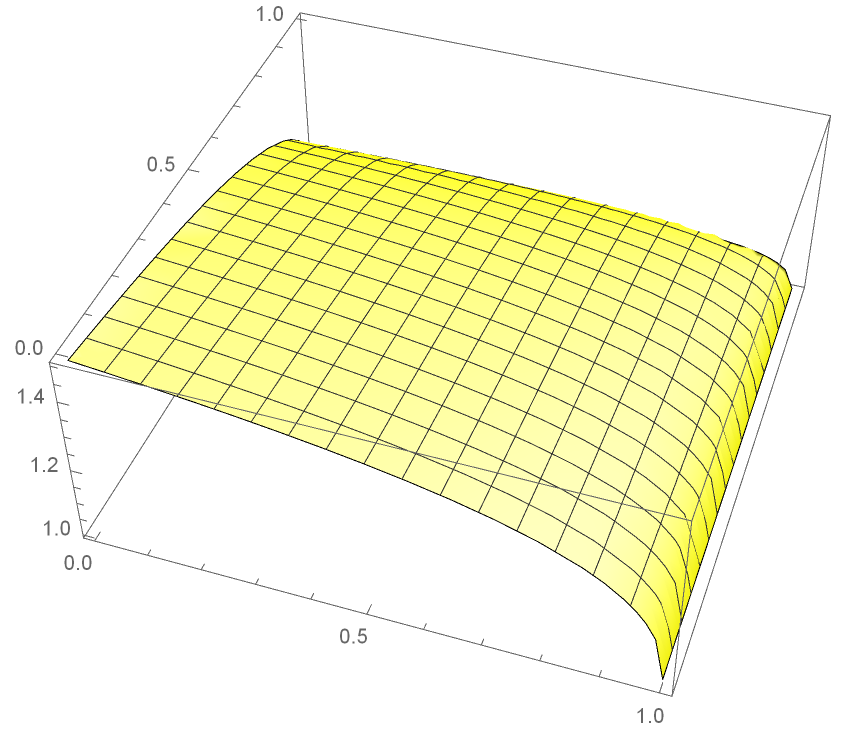}
    \caption{Graphical representation for $I_{12}(u_1,u_2)$ for Pareto distribution}
    \label{i12pareto}
\end{figure}
\begin{figure}
    \centering
    \includegraphics[width=0.5\linewidth]{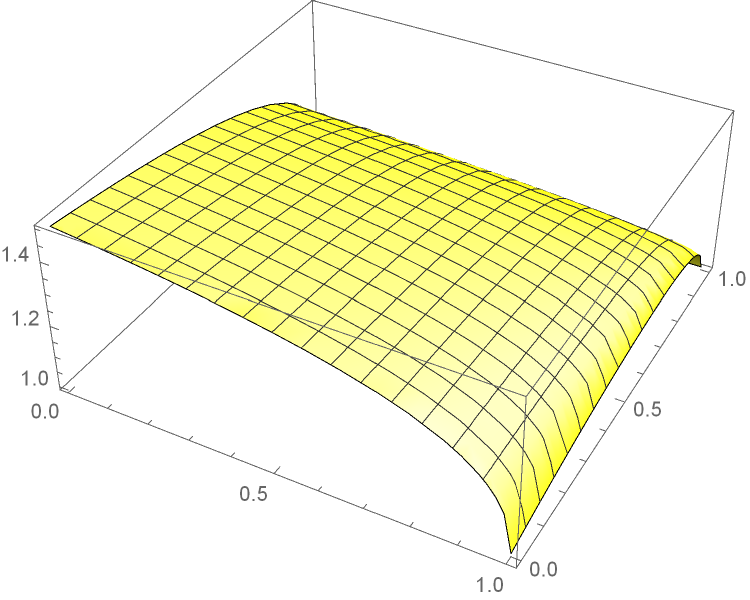}
    \caption{Graphical representation for $I_{21}(u_1,u_2)$ for Pareto distribution}
    \label{i21pareto}
\end{figure}
\end{ex}
\subsection{Properties of VBZC}
\begin{thm}
The bivariate distribution of $(X_1,X_2)$ is uniquely determined by
$\mathbf{I}(u_1,u_2)$.
\end{thm}
\begin{proof}
    We have 
    $J_{12}(u_1,u_2) = \int_0^{u_1} Q_{12}(p, u_2)\,dp,$ and hence $Q_{12}(u_1,u_2)=\frac{\partial J_{12}(u_1,u_2)}{\partial u_1}$. Using
$M^-_{12} = J_{12}/u_1$ and
$M^+_{12} = (\mu_{12}-J_{12})/(1-u_1)$
and substituting into the definition of $I_{12}$ yields \[
I_{12}(u_1,u_2)
=
\frac{\mu_{12}(u_2)u_1 - J_{12}(u_1,u_2)}
{\mu_{12}(u_2)u_1 - u_1J_{12}(u_1,u_2)}.
\]
Thus \[J_{12}(u_1,u_2)=\frac{\mu_{12}(u_2)\,u_1\,[1-I_{12}(u_1,u_2)]}
{1-u_1\,I_{12}(u_1,u_2)}.
\]
Differentiation with respect to $u_1$ yields $Q_{12}(u_1, u_2)$.\
Similarly\[J_{21}(u_1,u_2)=\frac{\mu_{21}(u_1)\,u_2\,[1-I_{21}(u_1,u_2)]}
{1-u_2\,I_{21}(u_1,u_2)},
\]
$Q_{21}$ is obtained from $I_{21}$.
Since $ \left(Q_{12}(u_1, u_2), Q_{21}(u_1, u_2)\right) $ uniquely determine the joint distribution, the result follows.
\end{proof}
\begin{thm}
For all $u_1,u_2\in(0,1)$,
\[
0 \le I_{12}(u_1,u_2) \le 1,
\qquad
0 \le I_{21}(u_1,u_2) \le 1 .
\]
Moreover,
\begin{enumerate}[(i)]
\item $I_{12}(u_1,u_2)=0$ for some $(u_1,u_2)$ if and only if the conditional
distribution of $X_1$ given $X_2>Q_2(u_2)$ is degenerate.
\item $I_{12}(u_1,u_2)=1$ for some $(u_1,u_2)$ if and only if
$M_{12}^-(u_1,u_2)=0$.
\end{enumerate}
Analogous statements hold for $I_{21}$.
\end{thm}

\begin{proof}
Since $X_1\ge 0$, the conditional quantile function $Q_{12}(p, u_2)\ge 0$ for all $p\in(0,1)$.  Hence $J_{12}(u_1,u_2)\ge 0$ and $\mu_{12}(u_2)\ge J_{12}(u_1,u_2)$, implying $M_{12}^-(u_1,u_2)\ge 0$ and $M_{12}^+(u_1,u_2)\ge 0$.

The function $Q_{12}(p, u_2)$ is non-decreasing in $p$.
Indeed, letting
\[
F_{X_1\mid X_2>t_2}(x_1)
=
P(X_1\le x_1\mid X_2>t_2),
\qquad t_2=Q_2(u_2),
\]
we have
\[
Q_{12}(p, u_2)=F_{X_1\mid X_2>t_2}^{-1}(p),
\]
the quantile function of a distribution function, which is non-decreasing in
$p$.
Therefore, for $p\le u_1$,
$Q_{12}(p, u_2)\le Q_{12}(u_1, u_2)$, and
\[
J_{12}(u_1,u
_2)
\le
u_1 Q_{12}(u_1, u_2).
\]
Similarly, for $p\ge u_1$,
$Q_{12}(p, u_2)\ge Q_{12}(u_1, u_2)$, and
\[
\mu_{12}(u_2)-J_{12}(u_1,u_2)
=
\int_{u_1}^1 Q_{12}(p, u_2)\,dp
\ge
(1-u_1)Q_{12}(u_1, u_2).
\]
Hence
\[
M_{12}^-(u_1,u_2)
\le
Q_{12}(u_1, u_2)
\le
M_{12}^+(u_1,u_2),
\]
which implies
\[
0\le \frac{M_{12}^-(u_1,u_2)}{M_{12}^+(u_1,u_2)} \le 1
\quad\text{and}\quad
0\le I_{12}(u_1,u_2)\le 1.
\]

\begin{enumerate}[(i)]
    \item $I_{12}(u_1,u_2)=0$ if and only if
$M_{12}^-(u_1,u_2)=M_{12}^+(u_1,u_2)$, which is equivalent to
$J_{12}(u_1,u_2)=u_1\mu_{12}(u_2)$.
Differentiating with respect to $u_1$ yields
$Q_{12}(u_1, u_2)=\mu_{12}(u_2)$, showing that the conditional distribution
of $X_1$ given $X_2>Q_2(u_2)$ is degenerate.
\item $I_{12}(u_1,u_2)=1$ if and only if $M_{12}^-(u_1,u_2)=0$, that is,
$J_{12}(u_1,u_2)=0$.
Since $Q_{12}(p, u_2)\ge 0$, this implies
$Q_{12}(p, u_2)=0$ for almost all $p\in[0,u_1]$, meaning that the lower
$u_1$ fraction of the conditional distribution has zero share.
\end{enumerate}

The proof for $I_{21}$ follows by interchanging the roles of the indices.
\end{proof}
\begin{thm}
Let $Y=(Y_1,Y_2)$ be defined by  $Y_1=a_1X_1$ and $Y_2=a_2X_2$ with $a_1,a_2>0$. Then, for all
$u_1,u_2\in(0,1)$,
\[
I_{12}^Y(u_1,u_2)=I_{12}^X(u_1,u_2),
\qquad
I_{21}^Y(u_1,u_2)=I_{21}^X(u_1,u_2).
\]
\end{thm}

\begin{proof}
Under the scaling transformation,
\[
Q_{12}^Y(u_1, u_2)=a_1 Q_{12}^X(u_1, u_2),
\qquad
\mu_{12}^Y(u_2)=a_1 \mu_{12}^X(u_2),
\]
and hence
\[
J_{12}^Y(u_1,u_2)
=
\int_0^{u_1} Q_{12}^Y(p\mid u_2)\,dp
=
a_1 J_{12}^X(u_1,u_2).
\]
Therefore,
\[
M_{12}^{Y,-}(u_1,u_2)
=
\frac{1}{u_1}J_{12}^Y(u_1,u_2)
=
a_1 M_{12}^{X,-}(u_1,u_2),
\]
and
\[
M_{12}^{Y,+}(u_1,u_2)
=
\frac{1}{1-u_1}\bigl(\mu_{12}^Y(u_2)-J_{12}^Y(u_1,u_2)\bigr)
=
a_1 M_{12}^{X,+}(u_1,u_2).
\]
Substituting into the definition
\[
I_{12}(u_1,u_2)
=
1-\frac{M_{12}^-(u_1,u_2)}{M_{12}^+(u_1,u_2)}
\]
yields $I_{12}^Y(u_1,u_2)=I_{12}^X(u_1,u_2)$. The proof for $I_{21}$ follows
analogously by interchanging the roles of the indices.
\end{proof}
\subsection{Nonparametric Estimators of the VBZC}

Let $(X_{1i},X_{2i})_{i=1}^n$ be an i.i.d. samples from an absolutely continuous 
bivariate distribution with finite first moments.
Fix $u_2 \in (0,1)$. Let $\hat F_2$ denote the empirical distribution function of $X_2$, and define the empirical marginal quantile
$\hat Q_2(u_2)=\inf \{ x : \hat F_2(x) \ge u_2 \}.$
Define the upper-tail subsample $S_{u_2}
=\{ i : X_{2i} > \hat Q_2(u_2) \}$ and let $n_{u_2}$ be the size of the subsample.
Let $X_{1(1)}^{u_2} \le \cdots \le X_{1(n_{u_2})}^{u_2}$ be the order statistics of $\{X_{1i} : i \in S_{u_2}\}$.
The empirical conditional quantile function of $X_1$ given 
$X_2 > \hat Q_2(u_2)$ is
$\hat Q_{12}(p , u_2)
=
X_{1(\lceil n_{u_2} p \rceil)}^{u_2},
\qquad p \in (0,1).$
The empirical lower partial mean is defined by
$\hat M^-_{12}(u_1,u_2)
=
\frac{1}{u_1}
\int_0^{u_1}
\hat Q_{12}(p , u_2)\, dp.$
Substituting the empirical quantile function, we obtain
$\hat M^-_{12}(u_1,u_2)
=
\frac{1}{u_1}
\int_0^{u_1}
X_{1(\lceil n_{u_2} p \rceil)}^{u_2}
\, dp.$
Since the empirical quantile function is stepwise constant, 
let $k = \lfloor n_{u_2} u_1 \rfloor$. Then
\[
\hat M^-_{12}(u_1,u_2)
=
\frac{1}{u_1}
\left[
\frac{1}{n_{u_2}}
\sum_{i=1}^{k}
X_{1(i)}^{u_2}
+
\left(
u_1 - \frac{k}{n_{u_2}}
\right)
X_{1(k+1)}^{u_2}
\right].
\]
The empirical conditional mean of $(X_1 \mid X_2 > Q_2)$ is estimated by the sample average of $X_1$ over the corresponding conditional subsample, $\hat \mu_{12}(u_2)=\frac{1}{n_{u_2}}\sum_{i \in S_{u_2}} X_{1i}.$
The empirical upper partial mean is given by
\[
\hat M^+_{12}(u_1,u_2)
=
\frac{
\hat \mu_{12}(u_2)
-
u_1 \hat M^-_{12}(u_1,u_2)
}
{1-u_1}.
\]
The nonparametric estimator of the vector Zenga component is therefore
\[
\hat I_{12}(u_1,u_2)
=
1
-
\frac{
\hat M^-_{12}(u_1,u_2)
}{
\hat M^+_{12}(u_1,u_2)
}.
\]
The estimator $\hat I_{21}(u_1,u_2)$ is obtained analogously by 
interchanging the roles of $X_1$ and $X_2$.
\begin{thm}
Let $\{(X_{1i},X_{2i})\}_{i=1}^n$ be an i.i.d. sample from a bivariate distribution with absolutely continuous distribution $F$. 
Assume that $E|X_1|<\infty$ and $M^+_{12}(u_1,u_2)\neq 0$ for $(u_1,u_2)\in(0,1)^2$. 
Then the nonparametric estimator $\hat I_{12}(u_1,u_2)$ converges almost surely to the  bivariate Zenga curve $I_{12}(u_1,u_2)$, that is,
\[
\hat I_{12}(u_1,u_2) \xrightarrow{a.s.} I_{12}(u_1,u_2),
\qquad n\to\infty.
\]
An analogous result holds for $\hat I_{21}(u_1,u_2)$.
\end{thm}
\begin{proof}
    By the Glivenko–Cantelli theorem, the empirical distribution function $\hat F_2$ converges uniformly almost surely to $F_2$. Since $F_2$ is continuous, the empirical quantile satisfies
$\hat Q_2(u_2) \xrightarrow{a.s.} Q_2(u_2).$
Consider the empirical conditional distribution function
\[
\hat F_{12}(x_1\mid u_2)
=
\frac{\sum_{i=1}^n 1\{X_{1i}\le x_1,\, X_{2i}>\hat Q_2(u_2)\}}
{\sum_{i=1}^n 1\{X_{2i}>\hat Q_2(u_2)\}}.
\]
The numerator is an average of i.i.d. indicators that converges almost surely to $P(X_1 \le x_1, X_2 > Q_2(u_2))$ by the law of large numbers and the convergence of $\hat Q_2(u_2)$. The denominator converges almost surely to $P(X_2 > Q_2(u_2)) > 0$. Hence, by the continuity,
$\hat F_{12}(\cdot\mid u_2) \xrightarrow{a.s.} F_{12}(\cdot\mid u_2).$
Uniform convergence of distribution functions implies convergence of the corresponding conditional quantile functions at every continuity point $p\in(0,1)$,
$\hat Q_{12}(p, u_2) \xrightarrow{a.s.} Q_{12}(p, u_2).$ Hence, by the dominated convergence theorem,
$\hat M^-_{12}(u_1,u_2)
=
\frac{1}{u_1}\int_0^{u_1}\hat Q_{12}(p, u_2)\,dp
\xrightarrow{a.s.}
M^-_{12}(u_1,u_2).$

Moreover, the empirical conditional mean
$\hat\mu_{12}(u_2)
=
\frac{1}{n_{u_2}}\sum_{i\in S_{u_2}}X_{1i}$
converges almost surely to $\mu_{12}(u_2)$ by the law of large numbers. Consequently,
$\hat M^+_{12}(u_1,u_2)
\xrightarrow{a.s.}
M^+_{12}(u_1,u_2).$

Finally, applying the continuous mapping theorem to
$\hat I_{12}(u_1,u_2)
=
1-\frac{\hat M^-_{12}(u_1,u_2)}{\hat M^+_{12}(u_1,u_2)},$ 
and noting that $M^+_{12}(u_1,u_2) \neq 0$ by assumption, we obtain
\[
\hat I_{12}(u_1,u_2)\xrightarrow{a.s.} I_{12}(u_1,u_2).
\]
The proof for $\hat I_{21}(u_1,u_2)$ follows by an identical argument, swapping the roles of $X_1$ and $X_2$.
\end{proof}

\subsection{Simulation Study}
A simulation study was carried out to investigate the  performance of the proposed nonparametric estimators 
$\hat{I}_{12}(u_1, u_2)$ and $\hat{I}_{21}(u_1, u_2)$. 
Random samples were generated from a bivariate lognormal distribution to reflect the positive and skewed nature typically associated with the inequality of data. The estimators were evaluated at selected quantile levels $(u_1, u_2) = (0.3, 0.3), (0.5, 0.5), (0.7, 0.7)$, representing lower, middle, and upper regions of the distribution. For each configuration, Monte Carlo simulations were performed with 500 replications for sample sizes $n = 50, 100, 200, 500$. The performance of the estimators was assessed using the absolute bias and mean squared error (MSE) and the simulation results are presented in Table \ref{simulation}.

\begin{table}[ht]
\centering
\renewcommand{\arraystretch}{1.3} 
\setlength{\tabcolsep}{10pt} 
\caption{Estimates, Absolute Bias and MSE of $(\hat I_{12}, \hat I_{21})$ based on simulated data}
\label{simulation}
\begin{tabular}{ccccc}
\toprule
$(u_1,u_2)$ & $n$ & $(\hat I_{12}, \hat I_{21})$ & Bias & MSE \\
\midrule

(0.3,0.3) & 50  & (0.8265, 0.8274) & (0.0174, 0.0129) & (0.0022, 0.0020) \\
          & 100 & (0.8388, 0.8369) & (0.0017, 0.0051) & (0.0009, 0.0009) \\
          & 200 & (0.8370, 0.8375) & (0.0040, 0.0054) & (0.0005, 0.0005) \\
          & 500 & (0.8401, 0.8390) & (0.0017, 0.0021) & (0.0002, 0.0002) \\

\addlinespace

(0.5,0.5) & 50  & (0.7879, 0.7912) & (0.0238, 0.0218) & (0.0039, 0.0036) \\
          & 100 & (0.7984, 0.8008) & (0.0145, 0.0114) & (0.0019, 0.0015) \\
          & 200 & (0.8051, 0.8058) & (0.0072, 0.0039) & (0.0008, 0.0009) \\
          & 500 & (0.8069, 0.8090) & (0.0029, 0.0034) & (0.0003, 0.0003) \\

\addlinespace

(0.7,0.7) & 50  & (0.7568, 0.7577) & (0.0436, 0.0424) & (0.0087, 0.0074) \\
          & 100 & (0.7801, 0.7739) & (0.0208, 0.0279) & (0.0041, 0.0038) \\
          & 200 & (0.7895, 0.7889) & (0.0157, 0.0121) & (0.0018, 0.0018) \\
          & 500 & (0.7971, 0.7961) & (0.0077, 0.0074) & (0.0007, 0.0007) \\

\bottomrule
\end{tabular}
\end{table}
It is observed that both the absolute bias and MSE decrease consistently as the sample size increases for all considered values of $(u_1,u_2)$. A similar pattern is evident for $\hat{I}_{21}(u_1,u_2)$ indicating that both estimators improve in accuracy with increasing sample size.

\section{Data analysis}
To illustrate the proposed nonparametric estimators of the bivariate Zenga curve, we consider cross-sectional data on digital infrastructure and digital capability for 23 countries in 2024. The data were obtained from the \href{https://www.itu.int}{International Telecommunication Union (ITU) database}, which provides internationally comparable indicators on information and communication technologies. Two indicators $X_1$ and $X_2$ were selected based on availability for all countries, $X_1$ represents the fixed broadband subscriptions per 100 inhabitants 
and $X_2$ gives the percentage of individuals possessing basic ICT skills.
The broadband variable reflects the physical infrastructure dimension of digital development, while the digital skills variable captures the human capital dimension. To facilitate the analysis, both variables were rescaled to the unit interval. The objective of the empirical analysis is to examine whether inequality in broadband access varies with the level of digital skills, and conversely, whether inequality in digital skills depends on broadband penetration.

The estimation of the bivariate Zenga surface is carried out using a nonparametric quantile-based approach. The empirical quantile function is constructed using order statistics. For each pair of quantile levels $(u_1, u_2)$, conditional subsamples are formed. Specifically, to estimate the component $I_{12}(u_1, u_2)$, we consider the subsample of broadband observations corresponding to countries whose digital skills exceed the $u_2$-th quantile. Similarly, to estimate $I_{21}(u_1, u_2)$, we consider the subsample of digital skill observations corresponding to countries whose broadband levels exceed the $u_1$-th quantile. For each conditional subsample, the lower and upper partial means are computed using the empirical quantile framework. The inequality measure is then obtained as one minus the ratio of the lower partial mean to the upper partial mean. To avoid instability due to the small sample size, the analysis is restricted to intermediate quantile levels, with $u_1$ and $u_2$ varying from $0.2$ to $0.8$. 
\begin{table}[ht]
\centering
\caption{Estimated values of bivariate Zenga measures at selected quantile combinations}
\begin{tabular}{|cccc|}

\hline
$u_1$ & $u_2$ & $\hat{I}_{12}(u_1,u_2)$ & $\hat{I}_{21}(u_1,u_2)$ \\
\hline
0.2 & 0.2 & 0.8282 & 0.7209 \\
0.5 & 0.2 & 0.7195 & 0.7318 \\
0.8 & 0.2 & 0.6220 & 0.6919 \\
0.2 & 0.5 & 0.6787 & 0.5660 \\
0.5 & 0.5 & 0.6596 & 0.6220 \\
0.8 & 0.5 & 0.5794 & 0.5421 \\
0.2 & 0.8 & 0.5533 & 0.4893 \\
0.5 & 0.8 & 0.6308 & 0.5064 \\
0.8 & 0.8 & 0.4080 & 0.3722 \\
\hline
\end{tabular}\label{zengavalues}
\end{table}
The estimated values of the bivariate Zenga measures for selected quantile combinations are presented in Table \ref{zengavalues}. It can be observed that the values of $\hat{I}_{12}(u_1, u_2)$ are highest at lower quantile levels, particularly at $(u_1, u_2) = (0.2, 0.2)$, where the inequality reaches 0.8282. This indicates substantial inequality in broadband access among countries with low levels of digital skills. As the quantile levels increase, the values of $\hat{I}_{12}$ generally decrease, reaching a minimum of 0.4080 at $(0.8, 0.8)$, suggesting a reduction in inequality among countries with higher digital capability. A similar trend is observed for $\hat{I}_{21}(u_1, u_2)$, although the magnitude is comparatively lower. The highest value of 0.7318 occurs at $(0.5, 0.2)$, indicating significant inequality in digital skills among countries with relatively low broadband penetration. As the quantile levels increase, the inequality measure decreases steadily, attaining a minimum value of 0.3722 at $(0.8, 0.8)$. The estimated bivariate Zenga surfaces are presented through contour plots in Figures \ref{contourI12} and \ref{contourI21}, which support the findings obtained from Table \ref{zengavalues}.
\begin{figure}
    \centering
    \includegraphics[width=0.5\linewidth]{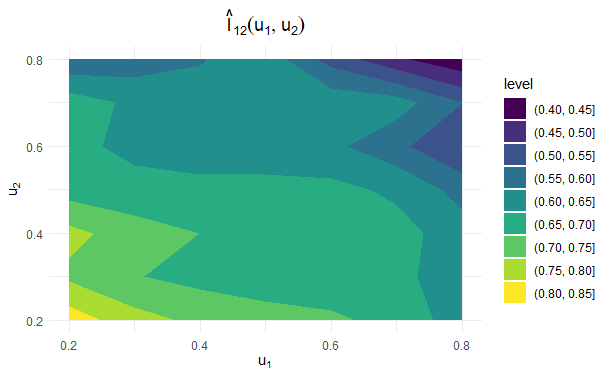}
    \caption{Contour plot of the estimated bivariate Zenga surface $\hat{I}_{12}(u_1, u_2)$, representing inequality in broadband access conditional on digital skills.}
    \label{contourI12}
\end{figure}
\begin{figure}
    \centering
    \includegraphics[width=0.5\linewidth]{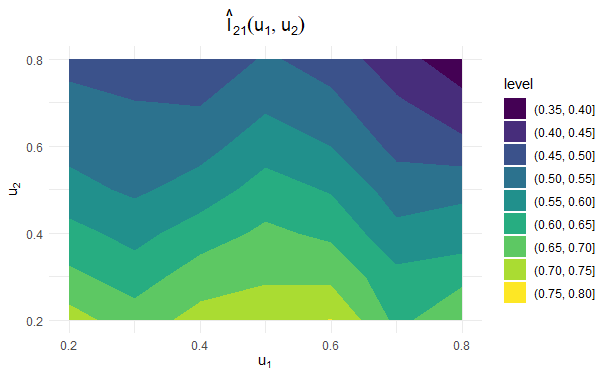}
    \caption{Contour plot of the estimated bivariate Zenga surface $\hat{I}_{21}(u_1, u_2)$, representing inequality in digital skills conditional on broadband access.}
    \label{contourI21}
\end{figure}
 The behavior of $\hat{I}_{12}(u_1, u_2)$ and $\hat{I}_{21}(u_1, u_2)$ is further examined through slice plots in Figures \ref{sliceI12} and \ref{sliceI21}. For fixed values of $u_2$, the measure $\hat{I}_{12}(u_1, u_2)$ generally decreases as $u_1$ increases, indicating a reduction in broadband inequality across higher quantiles. Similarly, for fixed values of $u_1$, the measure $\hat{I}_{21}(u_1, u_2)$ decreases as $u_2$ increases, suggesting that inequality in digital skills declines with increasing levels of broadband access.

\begin{figure}
    \centering
    \includegraphics[width=0.5\linewidth]{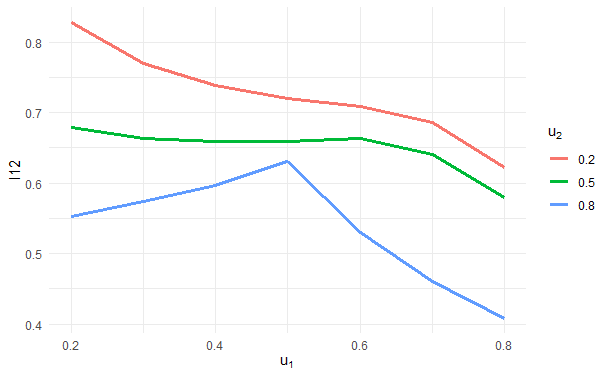}
    \caption{Slice plots of $\hat{I}_{12}(u_1, u_2)$ for selected values of $u_2$, illustrating the variation of broadband inequality across different quantile levels of digital skills.}
    \label{sliceI12}
\end{figure}
\begin{figure}
    \centering
    \includegraphics[width=0.5\linewidth]{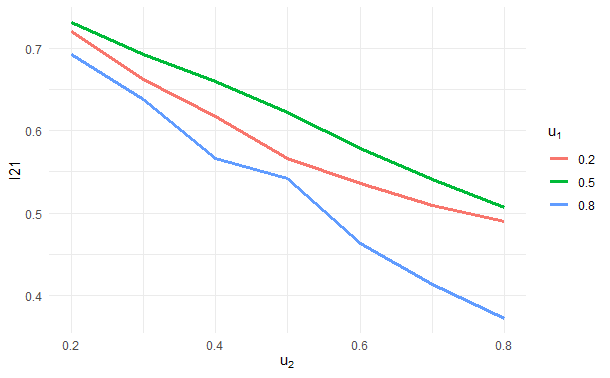}
    \caption{Slice plots of $\hat{I}_{21}(u_1, u_2)$ for selected values of $u_1$, showing the variation of digital skill inequality across different quantile levels of broadband access.}
    \label{sliceI21}
\end{figure}

These results indicate that inequality in broadband access is more pronounced among countries with lower levels of digital capability. As the level of digital skills increases, broadband inequality decreases, suggesting that improved digital capability is associated with a more equitable distribution of infrastructure. Likewise, inequality in digital skills is higher among countries with low broadband penetration and decreases steadily as broadband access improves.

An important observation from both the contour and slice plots is the asymmetric nature of the two measures. The surfaces $\hat{I}_{12}(u_1, u_2)$ and $\hat{I}_{21}(u_1, u_2)$ exhibit different patterns, indicating that inequality in one dimension conditional on another does not behave identically when the roles of the variables are reversed. This highlights the ability of the proposed bivariate Zenga measure to capture the directional dependence between the two variables.
\section*{Conflict of interest statement}
On behalf of all authors, the corresponding author states that there is no conflict of interest.
\section*{Data availability}
 The data used in this study are publicly available from the ITU database.
\bibliographystyle{apalike}
\bibliography{myref}

@article{lorenz1905methods,
  title={Methods of measuring the concentration of wealth},
  author={Lorenz, Max O},
  journal={Publications of the American Statistical Association},
  volume={9},
  number={70},
  pages={209--219},
  year={1905},
  publisher={Taylor \& Francis}
}

@article{gini1912variabilita,
  title={Variabilit{\`a} e mutabilit{\`a} (Variability and Mutability)},
  author={Gini, Corrado},
  journal={Tipografia di Paolo Cuppini, Bologna, Italy},
  volume={156},
  year={1912}
}

@inproceedings{zenga1990concentration,
  title={Concentration curves and concentration indexes derived from them},
  author={Zenga, Michele},
  booktitle={Income and Wealth Distribution, Inequality and Poverty: Proceedings of the Second International Conference on Income Distribution by Size: Generation, Distribution, Measurement and Applications, Held at the University of Pavia, Italy, September 28--30, 1989},
  pages={94--110},
  year={1990},
  organization={Springer}
}

@article{berti1995note,
  title={A note on {Z}enga concentration index},
  author={Berti, Patrizia and Rigo, Pietro},
  journal={Journal of the Italian Statistical Society},
  volume={4},
  number={3},
  pages={397--404},
  year={1995},
  publisher={Springer}
}

@article{greselin2010zenga,
  title={Zenga's New Index of Economic Inequality, Its Estimation, and an Analysis of Incomes in Italy},
  author={Greselin, Francesca and Pasquazzi, Leo and Zitikis, Ri{\v{c}}ardas},
  journal={Journal of Probability and Statistics},
  volume={2010},
  number={1},
  pages={718905},
  year={2010},
  publisher={Wiley Online Library}
}

@article{langel2012inference,
  title={Inference by linearization for {Z}enga’s new inequality index: a comparison with the {G}ini index},
  author={Langel, Matti and Till{\'e}, Yves},
  journal={Metrika},
  volume={75},
  number={8},
  pages={1093--1110},
  year={2012},
  publisher={Springer}
}

@article{polisicchio1993a,
  author  = {Polisicchio, M.},
  title   = {The Global {C}oncentration {M}easure and the {P}rinciple of {T}ransfers: {A}n {E}mpirical {A}nalysis ({I}talian)},
  journal = {Giornale degli Economisti e Annali di Economia},
  volume  = {52},
  pages   = {291--308},
  year    = {1993}
}

@article{polisicchio1993b,
  author  = {Polisicchio, M.},
  title   = {{S}ugli {O}rdinamenti {P}arziali {B}asati sulla {C}urva di {L}orenz e sulla {M}isura {P}untuale {Z}(p)},
  journal = {Quaderni di Statistica e Matematica Applicata alle Scienze Economico Sociali},
  volume  = {15},
  pages   = {63--86},
  year    = {1993}
}

@article{pollastri1987,
  author  = {Pollastri, A.},
  title   = {The {C}oncentration {C}urves {L}p and {Z}p in the {G}eneralized {L}ognormal {D}istribution},
  journal = {Giornale degli Economisti e Annali di Economia},
  volume  = {46},
  pages   = {639--663},
  year    = {1987}
}

@article{pr2024multivariate,
  title={Multivariate {L}eimkuhler Curve: Properties and Applications to Analysis of Bibliometric Data},
  author={Shifna,P R and Nair, N Unnikrishnan and Sunoj, S M},
  journal={Sankhya A},
  volume={86},
  number={2},
  pages={999--1024},
  year={2024},
  publisher={Springer}
}

@article{shifna2025extending,
  title={Extending the {L}orenz curve to higher dimensions: A multivariate quantile function approach},
  author={Shifna, P R and Sunoj, S M},
  journal = {Research Square},
  note    = {Preprint},
  year    = {2025}
}

@article{unnikrishnan2023properties,
  title={Properties of bivariate distributions represented through quantile functions},
  author={Nair, N Unnikrishnan and Vineshkumar, B},
  journal={American Journal of Mathematical and Management Sciences},
  volume={42},
  number={1},
  pages={1--12},
  year={2023},
  publisher={Taylor \& Francis}
}

@article{zenga1984proposta,
  title={Proposta per un indice di concentrazione basato sui rapporti fra quantili di popolazione e quantili di reddito},
  author={Zenga, Michele},
  journal={Giornale degli Economisti e Annali di Economia},
  pages={301--326},
  year={1984},
  publisher={JSTOR}
}

@article{zenga2007inequality,
  title={Inequality curve and inequality index based on the ratios between lower and upper arithmetic means},
  author={Zenga, Michele},
  journal={Statistica \& Applicazioni},
  volume={V},
  number={1},
  pages={3--27},
  year={2007},
  publisher={Vita e pensiero}
}

@article{jain1995multivariate,
  title={On multivariate weighted distributions},
  author={Jain, Kanchan and Nanda, Asok K},
  journal={Communications in Statistics-Theory and Methods},
  volume={24},
  number={10},
  pages={2517--2539},
  year={1995},
  publisher={Taylor \& Francis}
}

@article{nair2012some,
  title={Some properties of the new {Z}enga curve},
  author={Nair, N Unnikrishnan and Nair, K R Muraleedharan and Sreelakshmi, N},
  journal={Statistica \& Applicazioni},
  volume={X},
  number={1},
  pages={43--52},
  year={2012},
  publisher={Vita e pensiero}
}

@article{cruz2012digital,
  title={Digital divide across the {E}uropean Union},
  author={Cruz-Jesus, Frederico and Oliveira, Tiago and Bacao, Fernando},
  journal={Information \& Management},
  volume={49},
  number={6},
  pages={278--291},
  year={2012},
  publisher={Elsevier}
}

@article{billon2009disparities,
  title={Disparities in {ICT} adoption: A multidimensional approach to study the cross-country digital divide},
  author={Billon, Margarita and Marco, Rocio and Lera-Lopez, Fernando},
  journal={Telecommunications Policy},
  volume={33},
  number={10-11},
  pages={596--610},
  year={2009},
  publisher={Elsevier}
}

@article{chinn2007determinants,
  title={The determinants of the global digital divide: a cross-country analysis of computer and internet penetration},
  author={Chinn, Menzie D and Fairlie, Robert W},
  journal={Oxford Economic Papers},
  volume={59},
  number={1},
  pages={16--44},
  year={2007},
  publisher={Oxford University Press}
}

@article{thoene2025broadband,
  title={Broadband disparities and policy responses in {L}atin {A}merica and the {C}aribbean},
  author={Thoene, Ulf and Garc{\'\i}a Alonso, Roberto},
  journal={Digital Policy, Regulation and Governance},
  volume={27},
  number={5},
  pages={588--606},
  year={2025},
  publisher={Emerald Publishing Limited}
}

@article{vineshkumar2019bivariate,
  title={Bivariate quantile functions and their applications to reliability modelling},
  author={Vineshkumar, Balakrishnapillai and Nair, Narayanan Unnikrishnan},
  journal={Statistica},
  volume={79},
  number={1},
  pages={3--21},
  year={2019}
}
\end{document}